\documentclass[11pt]{article}
\usepackage{algorithm}
\usepackage{algorithmic}
\usepackage{graphicx}
\usepackage{amsmath}
\usepackage{amssymb}
\usepackage{color}
\usepackage{a4wide}
\newcommand{\bd}{\partial}
\newcommand{\D}{\mathcal{D}}

\newenvironment{proof}%
{\noindent\emph{Proof}.\hspace{1ex}}%
{\hfill\unitlength=0.18ex%
  \begin{picture}(12,12)
    \put(1,1){\framebox(9,9){}}
    \put(1,4){\framebox(6,6){}}
  \end{picture}\linebreak
}

\newtheorem{theorem}{Theorem}
\newtheorem{lemma}{Lemma}

\newtheorem{fact}{Fact}

\begin{document}

\title{A Note on Minimum-Sum Coverage by Aligned Disks}

\author{
Chan-Su Shin\thanks{School of Electronics and Information Engineering,  Hankuk~University~of~Foreign~Studies, Korea. Email: {\tt cssin@hufs.ac.kr}. Supported by National Research Foundation Grant(NRF2011-002827).}
}

\date{}
\maketitle

\begin{abstract}
  In this paper, we consider a facility location problem to find a minimum-sum coverage of $n$ points by disks centered at a fixed line. The cost of a disk with radius $r$ has a form of a non-decreasing function $f(r)= r^\alpha$ for any $\alpha \geq 1$. The goal is to find a set of disks under $L_p$ metric such that the disks are centered on the $x$-axis, their union covers the points, and the sum of the cost of the disks is minimized. Alt et al.~\cite{alt06} presented an algorithm running in $O(n^4\log n)$ time for any $\alpha > 1$ under any $L_p$ metric. We present a faster algorithm for this problem running in $O(n^2\log n)$ time for any $\alpha > 1$ and any $L_p$ metric.
\end{abstract}

\section{Introduction}

We consider a geometric facility location problem of finding $k$ disks whose union covers a set $P$ of $n$ input points with the minimum cost. A center of the disk of radius $r$ is often modeled as a base station(server) of transmission radius $r$ and the input points as point sensors(clients), so we assume the cost of the disk to be $r^\alpha$ for some real value $\alpha \geq 1$. The goal is to minimize $\sum_i r(D_i)^\alpha$ where the disks $D_i$ covering $P$ have radius $r(D_i)$. Specially, it is equivalent to minimizing the sum of their radii for $\alpha = 1$, and the sum of their areas for $\alpha = 2$.

Alt et al.~\cite{alt06} presented a number of results on several related problems. Among them, we focus on a restricted version in which the centers of the disks are aligned on a fixed line, simply saying $x$-axis. When the fixed line is not given but its orientation is fixed, finding the best line that guarantees the minimum coverage turned out to be quite hard even for $\alpha = 1$, thus they gave a PTAS approximation algorithm~\cite{alt06}.

Alt et al.~\cite{alt06} presented dynamic programming algorithms for this restricted problem of covering $P$ by aligned disks on a fixed line in time $O(n^2\log n)$ for $\alpha = 1$, and in time $O(n^4\log n)$ for any $\alpha > 1$ under any $L_p$ metric for $1\leq p < \infty$. For $L_\infty$ metric, they presented an $O(n^3\log n)$-time algorithm.

We present new dynamic programming algorithms running in $O(n^2\log n)$ time for any $\alpha > 1$ and any $L_p$ metric, and in $O(n^2)$ time for $L_\infty$ metric. Note here that the number of disks in the optimal covering is automatically determined in the algorithm. If one would want to restrict the number of disks used, say as a fixed $1\leq k \leq n$, then we can find at most $k$ disks whose union covers $P$ with minimum cost in a similar way. Actually we can find such $k$ disks for all $1\leq k\leq n$ in $O(n^3\log n)$ time in total.

\begin{theorem}
\label{thm:min}
Given a set $P$ of $n$ points in the plane and a non-decreasing cost function with $\alpha \geq 1$, we can compute an optimal disks centered on the $x$-axis such that their union covers $P$ and the sum of their costs is minimized in $O(n^2\log n)$ time for any fixed $L_p$ metric and in $O(n^2)$ time for $L_\infty$ metric.
\end{theorem}

\begin{theorem}
\label{thm:min_k}
Given a set $P$ of $n$ points in the plane and a non-decreasing cost function with $\alpha \geq 1$, we can compute a collection of optimal $k$ disks for $P$ such that $P$ is covered by at most $k$ disks for any $1\leq k \leq n$ and the sum of the costs of the disks is minimized in $O(n^3\log n)$ time for any fixed $L_p$ metric and in $O(n^3)$ time for $L_\infty$ metric.
\end{theorem}

\section{Geometric properties}

The formal definition of the problem is as follows: Given a set $P = \{p_1, \ldots, p_n\}$ of $n$ points in the plane, a real value $\alpha \geq 1$ and $L_p$ metric for some $p\geq 1$, find optimal disks $D_1,D_2,\ldots, D_k$ centered on the $x$-axis with radii $r(D_i)$ such that their union covers $P$ and the sum of their costs $\sum_{i} r^\alpha(D_i)$ is minimized.

As mentioned in~\cite{alt06}, we assume that all points in $P$ lie above or on the $x$-axis and no two points have the same $x$-coordinates. If a point $p$ is below the $x$-axis, we replace it with a new point $p'$ mirroring $p$ with respect to the $x$-axis, then we still have the same optimal covering. If $p$ is directly above $p'$, then any disk containing $p$ always contains $p'$, so we can simply discard $p'$ from $P$. Finally we assume the points of $P$ are in general position, i.e., no three or more points lie on the boundary of a disk centered on the $x$-axis.

We also notice that the optimal covering is not unique, so we assign the lexicographic order to such optimal coverings, according to $x$-coordinates of their centers. Then we consider only the leftmost optimal covering $\D = \{D_1, \ldots, D_k\}$ with centers in increasing order on the $x$-axis.

For a while, let us consider $L_p$ metric only for $1\leq p <\infty$. Let $\bd R$ denote the boundary of a closed region $R$. We denote by $t_i$ the highest point(or apex) of $\bd D_i$, and by $a_i$ and $b_i$ the left and right intersection points of $\bd D_i$ with the $x$-axis, respectively. Let $B$ be the union of disks in $\D$. Then the following facts hold; the first one is mentioned also in~\cite{alt06}.

\begin{fact}\label{fact:apex}\rm{\cite{alt06}}
For each $1\leq i \leq k$, the apex $t_i$ of $D_i$ appears on $\bd B$.
\end{fact}

Let us consider $\bd D_i\cap \bd B$, i.e., the circular arc of $\bd D_i$ which appears on $\bd B$. By Fact~\ref{fact:apex}, $t_i$ must be contained on the arc, so the arc is divided into the left and right subarcs at $t_i$. Then we have the following fact.

\begin{fact}\label{fact:partition}
For each $1\leq i \leq k$, $\bd D_i\cap \bd B$ must contain either one point of $P$ at the apex $t_i$ or two points of $P$, one on the left subarc and the other on the right subarc of $\bd D_i \cap \bd B$.
\end{fact}
\begin{proof}
It is obvious that there must be at least one point of $P$ on $\bd D_i\cap \bd B$. Otherwise we can shrink $D_i$ to get a smaller cost until $\bd D_i$ contains some point. Also if one of the left and right arcs has no points, then we can shrink $D_i$ while keeping the point on the one subarc until some point lies either on the apex $t_i$ or on the other subarc containing no points. This contradicts to the optimality.
\end{proof}

Furthermore, we can also prove the following fact by a similar argument.
\begin{fact}\label{fact:noptsatintersect}
For any $1\leq i< k$, there is no point of $P$ which lies at the intersection $\bd D_i \cap \bd D_{i+1}$ if the intersection exists.
\end{fact}

For each $1\leq i < k$, we define $\ell_i$ as a vertical line between $D_i$ and $D_{i+1}$; if $D_i$ intersects $D_{i+1}$, then $\ell_i$ is a vertical line through intersections $\bd D_i\cap \bd D_{i+1}$, otherwise $\ell_i$ is an arbitrary vertical line in between $b_i$ and $a_{i+1}$. For convenience, we define $\ell_0$ as a vertical line in the left of $a_1$ and $\ell_{k}$ as a vertical line in the right of $b_k$. Then Fact~\ref{fact:noptsatintersect} implies that no points of $P$ lie on $\ell_i$ for any $0\leq i\leq k$.

Let $P_i\subseteq P$ be the points of $P$ lying between $\ell_{i-1}$ and $\ell_i$ for $1\leq i\leq k$. Then we know that $P_i$ contains at least one point by Fact~\ref{fact:partition}, and they are pairwise disjoint and their union is the same as the whole set $P$. Let $C_i$ be the smallest axis-centered disk containing $P_i$. Clearly $\{C_1,\ldots, C_k\}$ is a covering for $P$. We have the following lemma.

\begin{lemma}\label{lem:partition}
$\sum_{1\leq i\leq k} r^\alpha(C_i) = \sum_{1\leq i\leq k} r^\alpha(D_i)$.
\end{lemma}
\begin{proof}
Since $\{D_1,\ldots, D_k\}$ is the optimal covering for $P$, it holds that $\sum_i r^\alpha(D_i) \leq \sum_i r^\alpha(C_i)$. By the definition of $P_i$, we see that $P_i\subseteq P\cap D_i$ for each $1\leq i < k$. thus $r(C_i) \leq r(D_i)$. Since $f(r) = r^\alpha$ is a nondecreasing function for $\alpha \geq 1$, $\sum_i r^\alpha(C_i) \leq \sum_i r^\alpha(D_i)$, which completes the lemma.
\end{proof}

The above lemma means that there is a partition of $P$ into $P_i$'s separated by vertical lines such that the set of the smallest disks containing $P_i$'s is an optimal covering for $P$. Using this lemma, we can derive a fast dynamic programming algorithm.

\section{Dynamic programming algorithm}

Alt et al.~\cite{alt06} defined a \emph{pinned} disk(or circle) as the leftmost smallest axis-centered disks enclosing some fixed subset of points of $P$, so the pinned disk contains at least one point on its boundary. The disk $C_i$ defined in Lemma~\ref{lem:partition} is a pinned disk. It is obvious that the optimal covering $\D$ is a subset of such pinned disks.

In \cite{alt06}, the dynamic programming algorithm chooses pinned disks with minimum cost from all $O(n^2)$ pre-computed pinned disks by testing the coverage condition for possible pairs of pinned disks that no other points of $P$ lie outside the chosen pair. Investigating such pairs requires at least $\Omega(n^4)$ time. But Lemma~\ref{lem:partition} tells us there must be a partition $P_1, \ldots, P_k$, separated by vertical lines, such that a set of the smallest disks $C_i$ enclosing $P_i$ is indeed an optimal covering for $P$. Thus we simply go through the input points in some linear order, not through the pinned disks, and compute $C_i$ by implicitly using the farthest Voronoi diagram to check the coverage condition.

Let $A$ be an array in which $A[i]$ stores the minimum cost for a subset $\{p_i, p_{i+1}, \ldots, p_n\}$. The minimum cost for the whole set $\{p_1,\ldots, p_n\}$ will be stored at $A[1]$. If we denote by $D(\{p_i,\ldots,p_j\})$ the smallest disk enclosing $\{p_i,\ldots,p_j\}$, then we have the following recurrence relation:

\begin{equation*}
A[i]=\begin{cases}
0 & \text{if $i > n$},\\
\min_{i\leq j\leq n} \{r^\alpha(D(\{p_i,\ldots,p_j\})) + A[j+1] \} & \text{if $1\leq i\leq n$}.
\end{cases}
\end{equation*}

The key step is to compute $D(\{p_i,\ldots,p_j\})$ efficiently. We can do this in amortized $O(\log n)$ time by maintaining the intersection of the $x$-axis with the farthest Voronoi diagram in a dynamic way. For a fixed $i$, $A[i]$ is computed in $O(n\log n)$ time, so the total time to compute $A[1]$ becomes $O(n^2\log n)$.

\begin{figure}
\begin{center}
    \includegraphics[width=10cm]{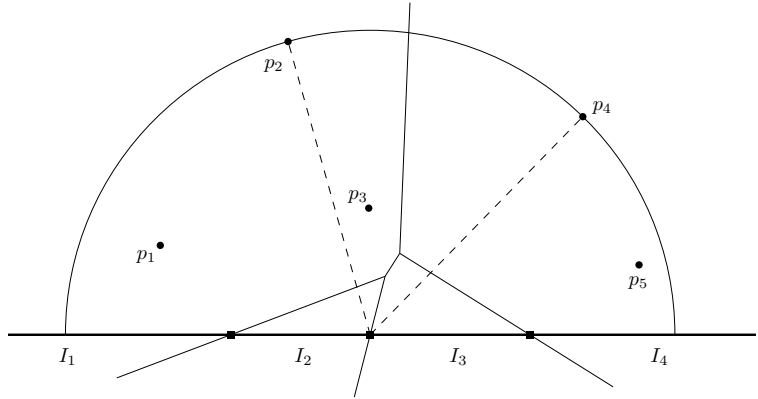}
\end{center}
\caption{The farthest Voronoi diagram of $P$ partitions the $x$-axis into intervals. Here we consider the $L_2$ metric. $p(I_1) = p_5$, $p(I_2)=p_4$, $p(I_3)=p_2$, $p(I_4) = p_1$, and $p_3$ has no its interval on $x$-axis.}
\label{fig:fvd}
\end{figure}

Let us focus on computing $A[i]$ for fixed $1\leq i\leq n$ in what follows.

As in Figure~\ref{fig:fvd}, the intersection of the farthest Voronoi diagram for $\{p_i, \ldots, p_j\}$ with the $x$-axis partitions the $x$-axis into intervals $I_1, I_2, \ldots, I_{\sigma(i,j)}$ from the left to the right, where $I_l$ is a half-open interval $I_l := [x_{l-1}, x_{l})$, where $x_0 = -\infty$ and $x_{\sigma(i,j)} = +\infty$. Each interval $I_l$ is a collection of the points from which the farthest point of $\{p_i,\ldots, p_j\}$ is the same. We denote by $p(I_l)$ the farthest point from any $x\in I_l$. Then a disk centered at some point $x\in I_l$ and with radius $|xp(I_l)|$ encloses all the points of $\{p_i,\ldots, p_j\}$.

Let $D(I_l)$ be the smallest disk enclosing $\{p_i,\ldots,p_j\}$ whose center lies in $I_l$. We have two cases. For a case that $\bd D(I_l)$ has one point at its apex, the point is indeed $p(I_l)$ and the center of $D(I_l)$ has the same $x$-coordinate as that of $p(I_l)$. For the other case, $\bd D(I_l)$ should have two points, i.e., $p(I_{l-1})$ and $p(I_l)$, so the center of $D(I_l)$ must be on $x_{l-1}$, i.e., the left endpoint of $I_l$. In either case, we can compute $D(I_l)$ and its radius in constant time, provided that the interval set is given.

We store the intervals in any balanced search tree $T_i$ which is able to support the insertion and deletion in $O(\log n)$ time; refer to~\cite{compgeom}. Each leaf node of $T_i$ keeps an interval $I_l$ together with $D(I_l)$ and $r(D(I_l))$ from left to right. Each internal node of $T_i$ stores the minimum one among the radii in the leaves of its subtree. Then the radius stored at the root of $T_i$ is the radius of the smallest disk enclosing $\{p_i, \ldots, p_j\}$. We can also update the smallest radius in $O(\log n)$ time whenever an interval is inserted or deleted.

For a fixed $i$, we need to evaluate $r(D\{p_i,\ldots,p_j\})$ for all $i\leq j \leq n$ to determine $A[i]$. We do this by maintaining a set of intervals $I_1,\ldots, I_{\sigma(i,j)}$ for $i\leq j \leq n$ incrementally from $j = i$ to $j = n$. For $j = i$, $T_i$ consists of one interval $(-\infty, +\infty)$ for $p_i$. We then update $T_i$ by adding $p_j$ one by one from $j = i+1,\ldots, n$. We now explain how we update $T_i$ for $\{p_{i},\ldots, p_{j-1}\}$ when $p_j$ is inserted.

We know that the interval $J=[a, b)$ for $p_j$ must appear as the new leftmost interval because $p_j$ is the rightmost point among $p_{i},\ldots, p_j$. Thus, $a = -\infty$. When $J$ is inserted into $T_i$, several consecutive intervals from the leftmost, which are completely contained in $J$, should be removed from $T_i$ or some interval which is partially overlapped with $J$ should be replaced with a shorter one. To identify such intervals efficiently, we need the following properties on the interval set.

\begin{lemma}\label{lem:interval} For $\{p_i, \ldots, p_j\}$ under any $L_p$ metric, $1\leq p < \infty$, the interval set $I_1, \ldots, I_{\sigma(i,j)}$ has two properties: (1) if a point of $\{p_i, \ldots, p_j\}$ has its interval in the interval set, then the interval is connected, and (2) $p_x(I_1)> p_x(I_2) >\cdots >p_x(I_{\sigma(i,j)})$, where $p_x(I)$ is the $x$-coordinate of the point $p(I)$.
\end{lemma}
\begin{proof}
A bisector of two points in $\{p_i,\ldots, p_j\}$ under any $L_p$ metric intersects the $x$-axis because no two points have the same $x$-coordinates. The bisector is monotone to the $x$-axis and the $y$-axis~\cite{dtlee80}, so it intersects the $x$-axis exactly once. This fact immediately implies that an interval for any point is either empty or connected.
%
%To prove the connectedness, we suppose that $p_j$ has two disjoint intervals $I$ and $L$, where $I$ is to the left of $L$. There must be one or more intervals between them, denote by $J$ the interval to the right of $I$ and by $K$ the interval to the left of $L$. Note that $J$ is not necessarily different with $K$. Let $D$ be a smallest disk centered at $I\cap J$, i.e., the common endpoint of $I$ and $J$ which encloses all points in $\{p_i,\ldots, p_j\}$. Then $p(I)$ and $p(J)$ lie on $\bd D$. Similarly, let $D'$ be a smallest disk centered at $K\cap L$ enclosing all the points. Since $p(I) = p(L) = p_j$, they must be on one of two intersections $\bd D \cap \bd D'$, clearly the one above the $x$-axis. Also the lune $D\cap D'$ contains all the points in $\{p_i, \ldots, p_j\}$. This implies that $p(J)$ must lie on the right boundary arc of the lune. The bisector of $p(I)$ and $p(J)$ intersects the $x$-axis at $I\cap J$, thus the points on the $x$-axis to the left of $I\cap J$ is farther to $p(J)$ than to $p(I)$, which contradicts that $I$ is in the left of $J$.

For the second property, it suffices to show that $p_x(I) > p_x(I')$ for any two consecutive intervals $I$ and $I'$ in $\{I_1,\ldots,I_{\sigma(i,j)}\}$ where $I$ is in the left of $I'$. Consider the half-circle of the smallest disk centered at $I\cap I'$ above the $x$-axis. It passes through $p(I)$ and $p(I')$. Since the half-circle intersects with the bisector of $p(I)$ and $p(I')$ exactly once, $p(I)$ should be in the right of $p(I')$ along the half-circle. This means that $p_x(I) > p_x(I')$ because the half-circle is monotone to the $x$-axis.
\end{proof}

We now decide the position of $b$, the right endpoint of $J$. This is equivalent to finding some interval $I_l$ which will contain $b$ after $p_j$ is inserted. Then $b$ is defined as the intersection of $I_l$ with the bisector of $p_j$ and $p(I_l)$. By Lemma~\ref{lem:interval}, we can find $I_l$ by scanning the intervals from left to right and checking if it intersects with the bisector of $p_j$ and $p(I_l)$. Once $I_l$ is found, we (1) delete the intervals $I_1,\cdots, I_{l-1}$ from $T_i$ because they are completely contained in $J$, (2) insert $J$ for $p_j$ into $T_i$, and (3) replace(i.e., delete then insert) $I_l$ with $I_l\setminus J$. If some interval is removed from $T_i$, then it is never inserted again into $T_i$. Hence, for a fixed $i$, we can compute the smallest disks enclosing $\{p_i, \ldots, p_j\}$ for all $i\leq j \leq n$ in $O((n-i)\log n)$ time. In other words, we can compute $A[i]$ in $O((n-i)\log n) =O(n\log n)$ time. As a result, we can compute $A[1]$, starting from $A[n]$, in $O(n^2\log n)$ time in total. This requires $O(n)$ space. The detailed algorithm is summarized below.

\begin{algorithm}[htb]
  \caption{\textsc{MinCostAlignedCoverage}$(P, \alpha)$}
  \begin{algorithmic}[1]
    \renewcommand{\algorithmicrequire}{\textbf{Input:}}
    \REQUIRE A set $P$ of $n$ points $\{p_1,\ldots, p_n\}$ and $\alpha \geq 1$.
    \renewcommand{\algorithmicrequire}{\textbf{Output:}}
    \REQUIRE An optimal covering $\D = \{D_1,\ldots,D_k\}$ for $P$, and its coverage cost.
    \STATE $A[n+1] = 0$.
    \FOR  {$i \leftarrow n$ downto $1$}
        \STATE Initialize $T_i$ as one interval $(-\infty, +\infty)$ for $p_i$.
        \STATE Initialize $A[i] = r^\alpha(D(\{p_i\})) + A[i+1]$.
         \FOR {$j \leftarrow i+1$ upto $n$}
            \STATE Find the first interval $I_l$ such that the bisector $B$ of $p_j$ and $p(I_l)$
                intersects $I_l$ by scanning the intervals in $T_i$ from left to right.\\
            \STATE $J := [-\infty, B\cap I_l)$
            \STATE Remove intervals $I_1,\ldots, I_{l-1}$, replace $I_l$ with $I_l\setminus J$, and insert $J$ in $T_i$.
            \STATE Let $r$ be the radius stored at the roof of $T_i$, i.e., $r = r(D(\{p_i,\ldots,p_j\}))$.
            \STATE $A[i] = \min(A[i],~r^\alpha+A[j+1])$
            \STATE Record the index $j$ at which the minimum is achieved.
        \ENDFOR\ $j$
    \ENDFOR\ $i$
    \STATE Reconstruct the optimal covering $\D$ from the recorded indices.
    \RETURN $\D$ and $A[1]$
  \end{algorithmic}
\end{algorithm}

\paragraph{Algorithm for $L_\infty$ metric.} Under this metric, the disk is an axis-aligned square. As before, we consider only the leftmost optimal covering by the lexicographic order. We can easily see that Fact~\ref{fact:apex} and Fact~\ref{fact:partition} also hold for $L_\infty$ metric if the apex $t_i$ of the disk is defined as the upper and right corner of the disk. To apply Lemma~\ref{lem:partition}, we would partition $P$ by the vertical lines through the intersection of two consecutive squares; in this metric they are the vertical lines containing the right sides of the squares. But some points can be on such vertical lines, thus Fact~\ref{fact:noptsatintersect} does not hold. Instead we choose the separation lines as the vertical lines a bit right the right sides of the squares. Then we can easily prove that Lemma~\ref{lem:partition} still holds. We now compute $A[i]$ similarly. The key step is to compute the smallest square $C$ enclosing $\{p_i, \ldots, p_j\}$ quickly. This square $C$ is determined by two points; $p_j$ and the highest point of $\{p_i,\ldots, p_{j-1}\}$, which can be computed in $O(1)$ time if we maintain the highest point during the incremental evaluation. Thus we can compute $A[i]$ in $O(n)$ time. The total time is $O(n^2)$, which completes the proof of Theorem~\ref{thm:min}.

\paragraph{Algorithm for fixed $k$ optimal disks.}
We can also consider the case when the maximum number of disks used to cover $P$ is given as a fixed value $1\leq k\leq n$. This can be similarly solved by filling a two dimensional table $A[i][k]$, the minimum cost needed to cover $p_i, \ldots, p_n$ with at most $k$ disks, in $O(kn^2\log n)$ total time. Actually we can find all optimal coverings for any $1\leq k\leq n$ in the same time. This completes the proof of Theorem~\ref{thm:min_k}.

\section{Concluding Remarks}

We can consider other disk coverage problems with practical constraints such as the connectivity of the disks. Recently, Chambers et al.~\cite{chambers11} investigated a problem of assigning radii to a given set of points in the plane such that the resulting set of disks is connected and the sum of radii, i.e., $\alpha = 1$ is minimized. When we bring such connectivity constraint to our problem for $\alpha \geq 1$, we need to find a ``connected'' set of disks which optimally covers the input points. When $\alpha = 1$, the smallest disk enclosing all points is the optimal coverage. However, we can easily check for $\alpha > 1$ that there are input points such that their optimal coverings consist of infinitely many disks. Thus we should restrict the number of disks used to cover, say $1\leq k\leq n$. But we have no idea how hard this problem is for a fixed $k$. It would be challenging for small $k$ such as $k = 2, 3$.


\begin{thebibliography}{1}

\bibitem{alt06}
Helmut Alt, Esther~M. Arkin, Herv\'{e} Br\"{o}nnimann, Jeff Erickson,
  S\'{a}ndor~P. Fekete, Christian Knauer, Jonathan Lenchner, Joseph S.~B.
  Mitchell, and Kim Whittlesey.
\newblock Minimum-cost coverage of point sets by disks.
\newblock In {\em Proceedings of the twenty-second annual symposium on
  Computational geometry}, SCG '06, pages 449--458, New York, NY, USA, 2006.
  ACM.

\bibitem{chambers11}
Erin~Wolf Chambers, S\'{a}ndor~P. Fekete, Hella-Franziska Hoffmann, Dimitri
  Marinakis, Joseph S.~B. Mitchell, Venkatesh Srinivasan, Ulrike Stege, and Sue
  Whitesides.
\newblock Connecting a set of circles with minimum sum of radii.
\newblock In {\em Proceedings of the 12th international conference on
  Algorithms and data structures}, WADS'11, pages 183--194, Berlin, Heidelberg,
  2011. Springer-Verlag.

\bibitem{compgeom}
Mark de~Berg, Marc van Kreveld, Mark Overmars, and Otfried Schwarzkopf.
\newblock {\em Computational Geometry: Algorithms and Applications}.
\newblock Springer-Verlag, second edition, 2000.

\bibitem{dtlee80}
D.~T. Lee.
\newblock Two-dimensional voronoi diagrams in the $l_p$-metric.
\newblock {\em J. ACM}, 27:604--618, October 1980.

\end{thebibliography}
\end{document}